\documentclass[letterpaper,11pt]{article}

\usepackage{latexsym}
\usepackage[ruled,vlined,titlenotnumbered]{algorithm2e}
\usepackage{verbatim}

\usepackage{amsmath}
\usepackage{amssymb}
\usepackage{amsfonts}

\usepackage{graphics,graphicx}
\usepackage[usenames]{color}

\newcommand{\BLUE}[1]{{#1}}

\newtheorem{theorem}{Theorem}[section]
\newtheorem{lemma}[theorem]{Lemma}
\newtheorem{invariant}[theorem]{Invariant}
\newtheorem{claim}[theorem]{Claim}

\newtheorem{definition}[theorem]{Definition}

\newenvironment{proof}{{\bf Proof:\ }}{\hfill$\Box$\medskip}

\newcommand{\ignore}[1]{}
\newcommand{\remove}[1]{}
\newcommand{\reals}{\hbox{$\rlap{\rm I} \> \kern-.2mm{\rm R}$}}

\newcommand{\HEAPSELECT}{\mbox{{\tt Heap}-{\tt Select}}}
\newcommand{\SOFTSELECT}{\mbox{{\tt Soft}-{\tt Select}}}
\newcommand{\SOFTSELECTHEAPIFY}{\mbox{{\tt Soft}-{\tt Select}-{\tt Heapify}}}
\newcommand{\MATSELECT}{\mbox{{\tt Mat}-{\tt Select}}}
\newcommand{\XYSELECT}{\mbox{{\tt X+Y}-{\tt Select}}}

\newcommand{\NUMROWS}{\mbox{{\sl num}-{\sl rows}}}
\newcommand{\JUMP}{\mbox{{\sl jump}}}
\newcommand{\SHIFT}{\mbox{{\sl shift}}}

\newcommand{\HEAP}{\mbox{{\tt heap}}}
\newcommand{\SOFTHEAP}{\mbox{{\tt soft}-{\tt heap}}}
\newcommand{\FINDMIN}{\mbox{{\tt find}-{\tt min}}}
\newcommand{\EXTRACTMIN}{\mbox{{\tt extract}-{\tt min}}}

\newcommand{\INSERT}{\mbox{\tt insert}}
\newcommand{\APPEND}{\mbox{\tt append}}
\newcommand{\SELECT}{\mbox{\tt select}}
\newcommand{\DELETE}{\mbox{\tt delete}}
\newcommand{\MELD}{\mbox{\tt meld}}

\newcommand{\LEFT}{\mbox{\sl left}}
\newcommand{\RIGHT}{\mbox{\sl right}}
\newcommand{\MIDDLE}{\mbox{\sl middle}}
\newcommand{\CORRUPT}{\mbox{\sl corrupt}}
\newcommand{\KEY}{\mbox{\sl key}}
\newcommand{\KEYP}{\mbox{\sl key}\,'}
\newcommand{\CHILD}{\mbox{\sl child}}

\newcommand{\NOT}{{\bf not}}

\begin{document}

\title{Selection from heaps, row-sorted matrices and $X+Y$ \\ using soft heaps
}

\author{Haim Kaplan \thanks{Blavatnik School of Computer Science, Tel Aviv University,
  Israel. Research supported by the
Israel Science Foundation grant no.\ 1841-14 and by a grant from the Blavatnik Computer Science Fund.
E-mail: {\tt haimk@post.tau.ac.il}.} \and
L\'{a}szl\'{o} Kozma \thanks{Department of Mathematics and Computer Science, TU Eindhoven, The Netherlands. E-mail:
{\tt Lkozma@gmail.com}.} \and
  Or Zamir \thanks{Blavatnik School of Computer Science, Tel Aviv University,
  Israel. E-mail:
{\tt orzamir@mail.tau.ac.il}.} \and  Uri Zwick \thanks{Blavatnik School of Computer Science, Tel Aviv University,
  Israel.  E-mail: {\tt zwick@tau.ac.il}.}}

\date{}

\maketitle

\begin{abstract}\noindent
We use \emph{soft heaps} to obtain simpler optimal algorithms for selecting the $k$-th smallest
item, and the set of~$k$ smallest items, from a heap-ordered tree, from a collection of sorted lists, and from $X+Y$, where $X$ and $Y$ are two unsorted sets.
Our results match, and in some ways extend and improve,
classical results of Frederickson (1993) and Frederickson and
Johnson (1982). In particular, for selecting the $k$-th smallest item,
or the set of~$k$ smallest items, from a collection of~$m$ sorted lists we obtain
a new optimal \emph{``output-sensitive"} algorithm that performs only $O(m+\sum_{i=1}^m \log(k_i+1))$
comparisons, where $k_i$ is the number of items of the $i$-th list that belong to the overall set of~$k$
smallest items.
\end{abstract}

\section{Introduction}\label{S-intro}

The input to the standard \emph{selection problem} is a set of $n$ items, drawn from a totally ordered domain, and an integer $1\le k\le n$. The goal is to return the $k$-th smallest item in the set. A classical result of Blum et al.\ \cite{BFPRT73} says that the selection problem can be solved deterministically in $O(n)$ time, 
i.e., faster than \emph{sorting} the set. The number of comparisons required for selection was reduced by Sch{\"{o}}nhage et al.\ \cite{SPP76} to $3n$, and by Dor and Zwick \cite{DoZw99,DoZw01} to about $2.95n$.

In the \emph{generalized selection} problem, we are also given a \emph{partial order} $P$ known to hold for the set of~$n$ input items. The goal is again to return the $k$-th smallest item. The corresponding \emph{generalized sorting} problem was extensively studied. It was shown by Kahn and Saks \cite{KaSa84} that the problem can be solved using only $O(\log e(P))$ comparisons, where~$e(P)$ is the number of \emph{linear extensions} of~$P$. Thus, the information-theoretic lower bound is tight for generalized sorting.
The algorithm of Kahn and Saks \cite{KaSa84} performs only $O(\log e(P))$ comparisons, but may spend much more time on deciding which comparisons to perform. Kahn and Kim \cite{KaKi95} and Cardinal et al.\ \cite{CFJJM13} gave algorithms that perform only $O(\log e(P))$ comparisons and run in polynomial time.

A moment's reflection shows that an algorithm that finds the $k$-th smallest item of a set, must also identify the set of~$k$ smallest items of the set.\footnote{\BLUE{The information gathered by a comparison-based algorithm corresponds to a partial order which can be represented by a DAG (Directed Acyclic Graph). Every \emph{topological sort} of the DAG corresponds to a total order of the items consistent with the partial order. Suppose that $e$ is claimed to be the $k$-th smallest item and suppose, for the sake of contradiction, that the set~$I$ of the items that are \emph{incomparable} with~$e$ is non-empty. Then, there is a topological sort in which $e$ is before all the items of~$I$, and another topological sort in which $e$ is after all the items of~$I$, contradicting the fact that $e$ is the $k$-th smallest item in all total orders consistent with the partial order.
}}
Given a partial order $P$, let $s_k(P)$ be the number of subsets of size~$k$ that may possibly be the set of~$k$ smallest items in~$P$. Then, $\log_2 s_k(P)$ is clearly a lower bound on the number of comparisons required to select the $k$-th smallest item, or the set of~$k$ smallest items. Unlike sorting, this information-theoretic lower bound for selection may be extremely weak. For example, the information-theoretic lower bound for selecting the \emph{minimum} is only $\log_2 n$, while $n-1$ comparisons are clearly needed (and are sufficient). To date, there is no characterization of pairs $(P,k)$ for which the information-theoretic lower bound for selection is tight, nor an alternative general technique to obtain a tight lower bound.

Frederickson and Johnson \cite{FrJo82,FrJo84,FrJo90} and Frederickson \cite{Frederickson93} studied the generalized selection problem for some interesting specific partial orders. Frederickson \cite{Frederickson93} considered the case in which the input items are items of a binary \emph{min-heap}, i.e., they are arranged in a binary tree, with each item smaller than its two children. Frederickson \cite{Frederickson93} gave a complicated algorithm that finds the $k$-th smallest item using only $O(k)$ comparisons, matching the information-theoretic lower bound for this case. (Each subtree of size~$k$ of the heap, containing the root, can correspond to the set of~$k$ smallest items, and there are $\frac{1}{k+1}{2k \choose k}$, the $k$-th \emph{Catalan} number, such subtrees.)

Frederickson and Johnson \cite{FrJo82,FrJo84,FrJo90} considered three other interesting special cases. (i) The input items are in a collection of sorted lists, or equivalently they reside in a row-sorted matrix; (ii) The input items reside in a collection of matrices, where each matrix is both row- and column-sorted; (iii) The input items are $X+Y$, where $X$ and $Y$ are unsorted sets of items.\footnote{By $X+Y$ we mean the set of pairwise sums $\{x+y ~|~ x\in X, y\in Y\}$.} For each of these cases, they present a selection algorithm that matches the information-theoretic lower bound.

We note in passing that \emph{sorting} $X+Y$ is a well studied problem. Fredman \cite{Fredman76b} showed that $X+Y$, where $|X|=|Y|=n$, can be sorted using only $O(n^2)$ comparisons, but it is not known how to do it in $O(n^2)$ time. (An intriguing situation!) Fredman \cite{Fredman76b} also showed that $\Omega(n^2)$ comparisons are required, if only comparisons between items in $X+Y$, i.e., comparisons of the form $x_i+y_j\le x_k+y_\ell$, are allowed. Lambert \cite{Lambert92} and Steiger and Streinu \cite{StSt95} gave algorithms that sort $X+Y$ in $O(n^2\log n)$ time using only $O(n^2)$ comparisons. Kane et al.\ \cite{KLM17}, in a breakthrough result, have shown recently that $X+Y$ can be sorted using only $O(n\log^2 n)$ comparisons of the form $(x_i+y_j)-(x_{i'}+y_{j'})\le (x_k+y_\ell)-(x_{k'}+y_{\ell'})$, but it is again not known how to implement their algorithm efficiently.

The \emph{median} of $X+Y$, on the other hand, can be found in $O(n\log n)$ time, and $O(n\log n)$ comparisons of items in $X+Y$, as was already shown by Johnson and Mizoguchi \cite{JoMi78} and Johnson and Kashdan \cite{JoKa78}. The selection problem from $X+Y$ becomes more challenging when $k=o(n^2)$.

Frederickson \cite{Frederickson93} gives two applications for selection from a binary min-heap. The first is in an algorithm for listing the~$k$ smallest spanning trees of an input graph. The second is a certain resource allocation problem. Eppstein \cite{Eppstein98} uses the heap selection algorithm in his $O(m+n\log n+k)$ algorithm for generating the~$k$ shortest paths between a pair of vertices in a digraph. As pointed out by Frederickson and Johnson \cite{FrJo84}, selection from $X+Y$ can be used to compute the Hodges-Lehmann \cite{HoLe63} estimator in statistics. Selection from a matrix with sorted rows solves the problem of ``optimum discrete distribution of effort" with concave functions, i.e., the problem of maximizing $\sum_{i=1}^m f_i(k_i)$ subject to $\sum_{i=1}^m k_i =k$, where the $f_i$'s are concave and the $k_i$'s are non-negative integers. (See Koopman \cite{Koopman53} and other references in \cite{FrJo84}.) Selection from a matrix with sorted rows is also used by Brodal et al.\ \cite{BFGL09} and Bremner et al.\ \cite{BCDEHILPT14}.

The $O(k)$ heap selection algorithm of Frederickson \cite{Frederickson93} is fairly complicated. The na\"ive algorithm for the problem runs in $O(k\log k)$ time. Frederickson first improves this to $O(k\log\log k)$, then to $O(k 3^{\log^*k})$, then to $O(k 2^{\log^*k})$, and finally to $O(k)$.

Our first result is a very simple $O(k)$ heap selection algorithm obtained by running the na\"ive $O(k\log k)$ algorithm using an auxiliary \emph{soft heap} instead of a standard heap. Soft heaps, discussed below, are fairly simple data structures whose implementation is not much more complicated than the implementation of standard heaps. Our overall $O(k)$ algorithm is thus simple and easy to implement and comprehend.

Relying on our simple $O(k)$ heap selection algorithm, we obtain simplified algorithms for selection from row-sorted matrices and from $X+Y$. Selecting the $k$-th item from a row-sorted matrix with~$m$ rows using our algorithms requires $O(m+k)$ time, if $k\le 2m$, and $O(m\log\frac{k}{m})$ time, if $k\ge 2m$, matching the optimal results of Frederickson and Johnson \cite{FrJo82}. Furthermore, we obtain a new optimal ``output-sensitive'' algorithm whose running time is $O(m+\sum_{i=1}^m \log(k_i+1))$, where~$k_i$ is the number of items of the $i$-th row that belong to the set of~$k$ smallest items in the matrix.
We also use our simple $O(k)$ heap selection algorithm to obtain simple optimal algorithms for selection from $X+Y$.

Soft heaps are ``approximate'' priority queues introduced by Chazelle \cite{Chazelle00b}. They form a major building block in his deterministic $O(m\alpha(m,n))$-time algorithm for finding minimum spanning trees \cite{Chazelle00a}, which is currently the fastest known deterministic algorithm for the problem. Chazelle \cite{Chazelle00b} also shows that soft heaps can be used to obtain a simple linear time (standard) selection algorithm. (See the next section.) Pettie and Ramachandran \cite{PeRa02} use soft heaps to obtain an \emph{optimal} deterministic minimum spanning algorithm with a yet unknown running time. A simplified implementation of soft heaps is given in Kaplan et al.\ \cite{KTZ13}.

All algorithms considered in the paper are \emph{comparison-based}, i.e., the only operations they perform on the input items are pairwise comparisons. In the selection from $X+Y$ problem, the algorithms make pairwise comparisons in~$X$, in~$Y$ and in~$X+Y$. The number of comparisons performed by the algorithms presented in this paper dominates the total running time of the algorithms.

The rest of the paper is organized as follows. In Section~\ref{S-soft-heaps} we review the definition of soft heaps. In Section~\ref{S-heap} we describe our heap selection algorithms. In Section~\ref{S-binary} we describe our basic algorithm for selection from \emph{binary} min-heaps. In Sections~\ref{S-dary} and~\ref{S-general} we extend the algorithm to $d$-ary heaps and then to general heap-ordered trees and forests. In Section~\ref{S-matrix} we describe our selection algorithms from row-sorted matrices. In Section~\ref{S-mk} we describe a simple $O(m+k)$ algorithm which is optimal if $k=O(m)$. In Section~\ref{S-mlogkm} we build on the $O(m+k)$ algorithm to obtain an optimal $O(m\log\frac{k}{m})$ algorithm, for $k\ge 2m$. In Sections~\ref{S-Slogni} and~\ref{S-Slogki} we obtain new results that were not obtained by Frederickson and Johnson \cite{FrJo82}. In Section~\ref{S-Slogni} we obtain an $O(m+\sum_{i=1}^m \log n_i)$ algorithm, where $n_i\ge 1$ is the length of the $i$-th row of the matrix. In Section~\ref{S-Slogki} we obtain the new $O(m+\sum_{i=1}^m \log(k_i+1))$ optimal output-sensitive algorithm. In Section~\ref{S-XY} we present our selection algorithms from $X+Y$.
In Section~\ref{S-XY-mnk} we give a simple $O(m+n+k)$ algorithm, where $|X|=m$, $|Y|=n$. In Section~\ref{S-XY-mlogmk} we give a simple $O(m\log\frac{k}{m})$ algorithm, for $m\ge n$ and $k\ge 6m$.
We conclude in Section~\ref{S-concl} with some remarks and open problems.

\section{Soft heaps}\label{S-soft-heaps}

Soft heaps, invented by Chazelle~\cite{Chazelle00b}, support the following operations:

$\;\;\;\SOFTHEAP(\varepsilon)$: Create and return a new, empty soft heap with \emph{error parameter}~$\varepsilon$.

$\;\;\;\INSERT(Q,e)$: Insert item~$e$ into soft heap $Q$.

$\;\;\;\MELD(Q_1, Q_2)$: Return a soft heap containing all items in heaps~$Q_1$ and~$Q_2$, destroying~$Q_1$ and~$Q_2$.

$\;\;\;\EXTRACTMIN(Q)$: Delete from the soft heap and return an item of minimum key in heap~$Q$.

In Chazelle \cite{Chazelle00b}, \EXTRACTMIN\ operations are broken into \FINDMIN\ and \DELETE\ operations. We only need combined \EXTRACTMIN\ operations. We also do not need \MELD\ operations in this paper.

The main difference between soft heaps and regular heaps is that soft heaps are allowed to \emph{increase} the keys of some of the items in the heap by an arbitrary amount. Such items are said to become \emph{corrupt}. The soft heap implementation chooses which items to corrupt, and by how much to increase their keys. The only constraint is that for a certain \emph{error parameter} $0\le \varepsilon<1$, the number of corrupt items \emph{in} the heap is at most $\varepsilon I$, where~$I$ is the number of \emph{insertions} performed so far. The ability to corrupt items allows the implementation of soft heaps to use what Chazelle~\cite{Chazelle00b} calls the ``data structures equivalent of car pooling'' to reduce the amortized time per operation to $O(\log\frac{1}{\varepsilon})$, which is the optimal possible dependency on~$\varepsilon$. In the implementation of Kaplan et al.~\cite{KTZ13}, \EXTRACTMIN\ operations take $O(\log\frac{1}{\varepsilon})$ amortized time, while all other operations take $O(1)$ amortized time. (In the implementation of Chazelle~\cite{Chazelle00b}, \INSERT\ operations take $O(\log\frac{1}{\varepsilon})$ amortized time while the other operations take $O(1)$ time.)

An \EXTRACTMIN\ operation returns an item whose current, possibly corrupt, key is the smallest in the heap. Ties are broken arbitrarily. (Soft heaps usually give many items the same corrupt key, even if initially all keys are distinct.) Each item~$e$ thus has two keys associated with it: its original key $e.\KEY$, and its current key in the soft heap $e.\KEYP$, where $e.key\le e.\KEYP$. If $e.\KEY<e.\KEYP$, then~$e$ is corrupt. The current key of an item may increase several times.

At first sight, it may seem that the guarantees provided by soft heaps are extremely weak. The only bound available is on the number of corrupt items currently \emph{in} the heap. In particular, \emph{all} items extracted from the heap may be corrupt. Nonetheless, soft heaps prove to be an extremely useful data structure. In particular, they play a key role in the fastest known deterministic algorithm of Chazelle \cite{Chazelle00a} for finding minimum spanning trees.

Soft heaps can also be used, as shown by Chazelle~\cite{Chazelle00b}, to select an \emph{approximate median} of $n$ items. Initialize a soft heap with some error parameter $\varepsilon<\frac{1}{2}$. Insert the~$n$ items into the soft heap and then perform $(1-\varepsilon)\frac{n}{2}$ \EXTRACTMIN\ operations. Find the maximum item~$e$, with respect to the original keys, among the extracted items. The rank of~$e$ is between $(1-\varepsilon)\frac{n}{2}$ and $(1+\varepsilon)\frac{n}{2}$. The rank is at least $(1-\varepsilon)\frac{n}{2}$ as~$e$ is the largest among $(1-\varepsilon)\frac{n}{2}$ items. The rank is at most $(1+\varepsilon)\frac{n}{2}$ as the items remaining in the soft heap that are smaller than~$e$ must be corrupt, so there are at most $\varepsilon n$ such items. For, say, $\varepsilon=\frac{1}{4}$, the running time of the algorithm is $O(n)$.

Using a linear time approximate median algorithm, we can easily obtain a linear time algorithm for selecting the $k$-th smallest item. We first compute the true rank~$r$ of the approximate median~$e$. If $r=k$ we are done. If $r>k$, we throw away all items larger than~$e$. Otherwise, we throw away all items smaller than~$e$ and replace~$k$ by $k-r$. We then continue recursively. In $O(n)$ time, we reduced $n$ to at most $(\frac{1}{2}+\varepsilon)n$, so the total running time is $O(n)$.

In this paper, we show the usefulness of soft heaps in solving generalized selection problems. We obtain simpler algorithms than those known before, and some results that were not known before.

In Chazelle~\cite{Chazelle00b} and Kaplan et al.~\cite{KTZ13}, soft heaps may corrupt items while performing any type of operation. It is easy, however, to slightly change the implementation of \cite{KTZ13} such that corruptions only occur \emph{following} \EXTRACTMIN\ operations. In particular, \INSERT\ operations do not cause corruption, and an \EXTRACTMIN\ operation returns an item with a smallest current key at the \emph{beginning} of the operation. These assumptions simplify algorithms that use soft heaps, and further simplify their analysis. The changes needed in the implementation of soft heaps to meet these assumptions are minimal. The operations \INSERT\ (and \MELD) are simply implemented in a \emph{lazy} way. The implementation of \cite{KTZ13} already has the property that \EXTRACTMIN\ operations cause corruptions only \emph{after} extracting an item with minimum current key.

We assume that an \EXTRACTMIN\ operation returns a pair $(e,C)$, where~$e$ is the extracted item, and~$C$ is a list of items that became corrupt after the extraction of~$e$, i.e., items that were not corrupt before the operation, but are corrupt after it. We also assume that $e.\CORRUPT$ is a bit that says whether~$e$ is corrupt. (Note that $e.\CORRUPT$ is simply a shorthand for $e.\KEY<e.\KEYP$.) It is again easy to change the implementation of \cite{KTZ13} so that \EXTRACTMIN\ operations return a list~$C$ of newly corrupt items, without affecting the amortized running times of the various operations. (In particular, the amortized running time of an \EXTRACTMIN\ operation is still $O(\log\frac{1}{\varepsilon})$, independent of the length of~$C$. As each item becomes corrupt only once, it is easy to charge the cost of adding an item to~$C$ to its insertion into the heap.)

\BLUE{
We stress that the assumptions we make on soft heaps in this paper can be met by minor and straightforward modifications of the implementation of Kaplan et al.~\cite{KTZ13}, as sketched above. No complexities are hidden here. We further believe that due to their usefulness, these assumptions will become the standard assumptions regarding soft heaps.}

\section{Selection from heap-ordered trees}\label{S-heap}

In Section~\ref{S-binary} we present our simple, soft heap-based, $O(k)$ algorithm for selecting the $k$-th smallest item, and the set of~$k$ smallest items from a binary min-heap. This algorithm is the cornerstone of this paper. For simplicity, we assume throughout this section that the input heap is infinite. In particular, each item~$e$ in the input heap has two children $e.\LEFT$ and $e.\RIGHT$. (A non-existent child is represented by a dummy item with key $+\infty$.) In Section~\ref{S-dary} we adapt the algorithm to work for $d$-ary heaps, for $d\ge 3$, using ``on-the-fly ternarization via heapification''. In Section~\ref{S-general} we extend the algorithm to work on any heap-ordered tree or forest. The results of Section~\ref{S-general} are new.

\subsection{Selection from binary heaps}\label{S-binary}

\begin{figure}
\begin{center}
\parbox{2.3in}{
\begin{algorithm}[H]
\DontPrintSemicolon
\NoCaptionOfAlgo
\caption{$\HEAPSELECT(r)$:}
\BLUE{$S\gets \emptyset$ \;}
$Q\gets \HEAP()$ \;
$\INSERT(Q,r)$ \;
\For{$i\gets 1$ \KwTo $k$}{
    $e\gets \EXTRACTMIN(Q)$ \;
    \BLUE{$\APPEND(S,e)$ \;}
    $\INSERT(Q,e.\LEFT)$ \;
    $\INSERT(Q,e.\RIGHT)$
}
\BLUE{\Return $S$}
\end{algorithm}
\caption{\label{A1}Extracting the $k$ smallest items from a binary min-heap with root~$r$ using a standard heap.}
}
\hspace*{1.5cm}
\parbox{3.2in}{
\begin{algorithm}[H]
\NoCaptionOfAlgo
\caption{$\SOFTSELECT(r)$:}
\DontPrintSemicolon
\BlankLine
\BLUE{$S\gets \emptyset$ \;}
$Q\gets \SOFTHEAP(1/4)$ \;
$\INSERT(Q,r)$ \;
\BLUE{$\APPEND(S,r)$ \;}
\BlankLine
\For{$i\gets 1$ \KwTo $k-1$}{
    $(e,C)\gets \EXTRACTMIN(Q)$ \;
    \lIf{\NOT\ e.\CORRUPT}{$C\gets C\cup\{e\}$}
    \For{$e\in C$}{
    $\INSERT(Q,e.\LEFT)$ \;
    $\INSERT(Q,e.\RIGHT)$ \;
    \BLUE{$\APPEND(S,e.\LEFT)$ \;
    $\APPEND(S,e.\RIGHT)$}
    }
}
\BLUE{\Return $\SELECT(S,k)$}
\end{algorithm}
\caption{\label{A2}Extracting the $k$ smallest items from a binary min-heap with root~$r$ using a soft heap.}
}
\end{center}
\end{figure}

The na\"ive algorithm for selection from a binary min-heap is given in Figure~\ref{A1}. The root $r$ of the input heap is inserted into an auxiliary heap (priority queue). The minimal item~$e$ is extracted from the heap and appended to a list~$S$. The two children of~$e$, if they exist, are inserted into the heap. This operation is repeated~$k$ times. After $k$ iterations, the items in~$S$ are the~$k$ smallest items in the input heap, in sorted order. Overall, $2k+1$ items are inserted into the heap and~$k$ items are extracted, so the total running time is $O(k\log k)$, which is optimal if the~$k$ smallest items are to be reported in sorted order.

Frederickson \cite{Frederickson93} devised a very complicated algorithm that outputs the~$k$ smallest items, not necessarily in sorted order, in only $O(k)$ time, matching the information-theoretic lower bound. In Figure~\ref{A2} we give our very simple algorithm $\SOFTSELECT(r)$ for the same task, which also runs in optimal $O(k)$ time and performs only $O(k)$ comparisons. Our algorithm is a simple modification of the na\"ive algorithm of Figure~\ref{A1} with the auxiliary heap replaced by a soft heap.
The resulting algorithm is much simpler than the algorithm of Frederickson \cite{Frederickson93}.

Algorithm $\SOFTSELECT(r)$ begins by initializing a soft heap~$Q$ with error parameter $\varepsilon=1/4$ and by inserting the root~$r$ of the input heap into it. Items inserted into the soft heap~$Q$ are also inserted into a list~$S$. The algorithm then performs $k-1$ iterations. In each iteration, the operation $(e,C)\gets\EXTRACTMIN$ extracts an item~$e$ with the smallest (possibly corrupt) key currently in~$Q$, and also returns the set of items $C$ that become corrupt as a result of the removal of~$e$ from~$Q$. If~$e$ is \emph{not} corrupt, then it is added to~$C$. Now, for each item $e\in C$, we insert its two children $e.\LEFT$ and $e.\RIGHT$ into the soft heap~$Q$ and the list~$S$.

Lemma~\ref{L-soft-1} below claims that $\SOFTSELECT(r)$ inserts the~$k$ smallest items of the input heap into the soft heap~$Q$. Lemma~\ref{L-soft-2} claims that, overall, only $O(k)$ items are inserted into~$Q$, and hence into~$S$. Thus, the $k$ smallest items in the input heap can be found by selecting the $k$ smallest items in the list~$S$ using a standard selection algorithm.

\begin{lemma}\label{L-soft-1} Algorithm $\SOFTSELECT(r)$ inserts the~$k$ smallest items from the input binary min-heap into the soft heap~$Q$. (Some of them may subsequently be extracted from the heap.)
\end{lemma}

\begin{figure}[t]
\begin{center}
\includegraphics[scale=0.5]{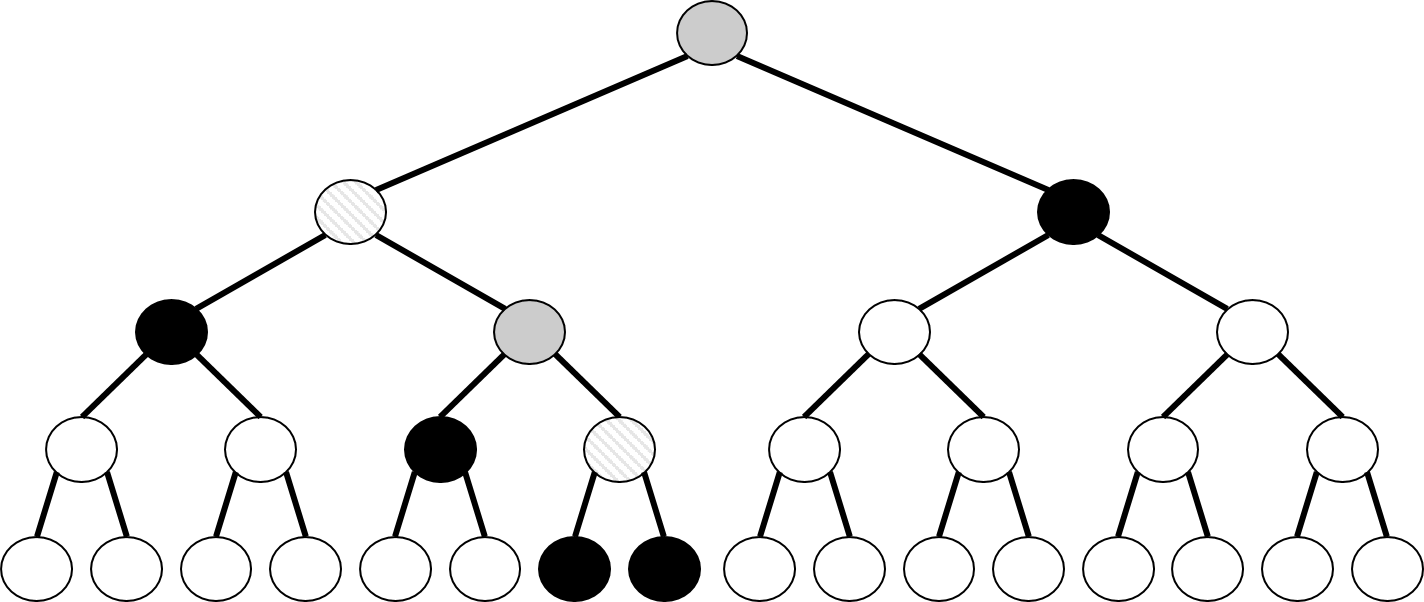}
\end{center}
\vspace*{-10pt}
\caption{Types of items in the input heap. White nodes belong to~$A$, i.e., were not inserted yet into the soft heap~$Q$; black nodes belong to the barrier~$B$; gray nodes belong to~$C$, i.e., are corrupt; striped nodes belong to~$D$, i.e., were already deleted.}\label{F-barrier}
\end{figure}

\begin{proof} At the beginning of an iteration of algorithm $\SOFTSELECT$, let $A$ be the set of items of the input binary heap that were not yet inserted into the soft heap~$Q$; let~$B$ be the set of items that were inserted, not yet removed and are not corrupt; let $C$ be the set of items that were inserted, not yet removed, and are \emph{corrupt}; let~$D$ be the set of items that were inserted and already \emph{deleted} from~$Q$. We prove below, by easy induction, the following two invariants:
\begin{itemize}
\setlength{\parsep}{0pt}
\setlength{\parskip}{0pt}
\setlength\itemsep{0pt}
 \item[(a)] All strict ancestors of items in~$B$ are in $C\cup D$.
 \item[(b)] Each item in~$A$ has an ancestor in~$B$.
\end{itemize}

Thus, the items in~$B$ form a \emph{barrier} that separates the items of~$A$, i.e., items that were not inserted yet into the heap, from the items of $C\cup D$, i.e., items that were inserted and are either corrupt or were already removed from the soft heap~$Q$. For an example, see Figure~\ref{F-barrier}.

Invariants (a) and (b) clearly hold at the beginning of the first iteration, when $B=\{r\}$ and the other sets are empty.
Assume that (a) and (b) hold at the beginning of some iteration. Each iteration removes an item from the soft heap. The item removed is either a corrupt item from~$C$, or an item (in fact the smallest item) on the barrier~$B$. Following the extraction, some items on the barrier~$B$ become corrupt and move to~$C$. The barrier is `mended' by inserting to~$Q$ the children of items on~$B$ that were extracted or became corrupt. By our simplifying assumption, insertions do not corrupt items, so the newly inserted items belong to~$B$ and are thus part of the new barrier, reestablishing (a) and (b).


We now make the following two additional claims:
\begin{itemize}
\setlength{\parsep}{0pt}
\setlength{\parskip}{0pt}
\setlength\itemsep{0pt}
 \item[(c)] The item extracted at each iteration is smaller than or equal to the smallest item on the barrier. (With respect to the original keys.)
 \item[(d)] The smallest item on the barrier cannot decrease.
\end{itemize}
Claim (c) follows immediately from the definition of an $\EXTRACTMIN$ operation and our assumption that corruption occurs only after an extraction. All the items on the barrier, and in particular the smallest item~$e$ on the barrier, are in the soft heap and are not corrupt. Thus, the extracted item is either~$e$, or a corrupt item $f$ whose corrupt key is still smaller than~$e$. As corruption can only increase keys, we have $f<e$.

Claim (d) clearly holds as items on the barrier at the end of an iteration were either on the barrier at the beginning of the iteration, or are children of items that were on the barrier at the beginning of the iteration.

Consider now the smallest item~$e$ on the barrier after $k-1$ iterations. As all extracted items are smaller than it, the rank of~$e$ is at least~$k$. Furthermore, all items smaller than~$e$ must be in $C\cup D$, i.e., inserted at some stage into the heap.
Indeed, let $f$ be an item of~$A$, i.e., an item not inserted into~$Q$. By invariant (b), $f$ has an ancestor $f'$ on the barrier. By heap order and the assumption that~$e$ is the smallest item on the barrier we indeed get $e\le f'<f$. Thus, the smallest~$k$ items were indeed inserted into the soft heap as claimed.
\end{proof}

The proof of Lemma~\ref{L-soft-1} relies on our assumption that corruptions in the soft heap occur only after \EXTRACTMIN\ operations. A slight change in the algorithm is needed if \INSERT\ operations may cause corruptions; we need to repeatedly add children of newly corrupt items until no new items become corrupt. (Lemma~\ref{L-soft-2} below shows that this process must end if $\varepsilon<\frac{1}{2}$. The process may \emph{not} end if $\varepsilon\ge \frac{1}{2}$.)
The algorithm, without any change, remains correct, and in particular Lemma~\ref{L-soft-1} holds, if \EXTRACTMIN\ operations are allowed to corrupt items before extracting an item of minimum (corrupt) key. The proof, however, becomes more complicated. (Claim (c), for example, does not hold in that case.)

\begin{lemma}\label{L-soft-2} Algorithm $\SOFTSELECT(r)$ inserts only $O(k)$ items into the soft heap~$Q$.
\end{lemma}

\begin{proof} Let $I$ be the number of insertions made by $\SOFTSELECT(r)$, and let $C$ be the number of items that become corrupt during the running of the algorithm. (Note that $\SOFTSELECT(r)$ clearly terminates.) Let $\varepsilon(=\frac{1}{4})$ be the error parameter of the soft heap. We have $I<2k+2C$, as each inserted item is either the root~$r$, or a child of an item extracted during one of the $k-1$ iterations of the algorithm, and there are at most $2k-1$ such insertions, or a child of a corrupt item, and there are exactly $2C$ such insertions. We also have $C<k+\varepsilon I$, as by the definition of soft heaps, at the end of the process at most $\varepsilon I$ items in the soft heap may be corrupt, and as only $k-1$ (possibly corrupt) items were removed from the soft heap. Combining these two inequalities we get $C<k+\varepsilon(2k+2C)$, and hence $(1-2\varepsilon)C<(1+2\varepsilon)k$. Thus, if $\varepsilon<\frac{1}{2}$ we get
\[ C \;<\; \frac{1+2\varepsilon}{1-2\varepsilon}\,k \quad,\quad I \;<\;2\left(1+\frac{1+2\varepsilon}{1-2\varepsilon}\right)k\;.\]
The number of insertions $I$ is therefore $O(k)$, as claimed. (For $\varepsilon=\frac{1}{4}$, $I<8k$.)
\end{proof}

Combining the two lemmas we easily get:

\begin{theorem}\label{T-soft} Algorithm $\SOFTSELECT(r)$ selects the~$k$ smallest items of a binary min-heap in $O(k)$ time.
\end{theorem}

\BLUE{
\begin{proof} The correctness of the algorithm follows from Lemmas~\ref{L-soft-1} and~\ref{L-soft-2}. Lemma~\ref{L-soft-2} also implies that only $O(k)$ operations are performed on the soft heap. As $\varepsilon=1/4$, each operation takes $O(1)$ amortized time. The total running time, and the number of comparisons, performed by the loop of $\SOFTSELECT(r)$ is thus $O(k)$. As the size of $S$ is $O(k)$, the selection of the smallest $k$ items from~$S$ also takes only $O(k)$ time.
\end{proof}
}

\subsection{Selection from $d$-ary heaps}\label{S-dary}

\BLUE{
Frederickson \cite{Frederickson93} claims, in the last sentence of his paper, that his algorithm for binary min-heaps can be modified to yield an optimal $O(dk)$ algorithm for $d$-ary min-heaps, for any $d\ge 2$, but no details are given. (In a $d$-ary heap, each node has (at most) $d$ children.)

We present two simple $O(dk)$ algorithms for selecting the $k$ smallest items from a $d$-ary min-heap. The first is a simple modification of the algorithm for the binary case. The second in a simple reduction from the $d$-ary case to the binary case.

Algorithm $\SOFTSELECT(r)$ of Figure~\ref{A2} can be easily adapted to work on $d$-ary heaps. We simply insert the $d$ children of an extracted item, or an item that becomes corrupt, into the soft heap. If we again let~$I$ be the number of items inserted into the sort heap, and~$C$ be the number of items that become corrupt, we get $I<d(k+C)$ and $C<k+\varepsilon I$, and hence
\[ C \;<\; \frac{1+d\varepsilon}{1-d\varepsilon}\,k \quad,\quad I \;<\;d\left(1+\frac{1+d\varepsilon}{1-d\varepsilon}\right)k\;,\]
provided that $\varepsilon<\frac{1}{d}$, e.g., $\varepsilon=\frac{1}{2d}$. The algorithm then performs $O(dk)$ \INSERT\ operations, each with an amortized cost of $O(1)$, and $k-1$ \EXTRACTMIN\ operations, each with an amortized cost of $O(\log\frac{1}{\varepsilon})=O(\log d)$. The total running time is therefore $O(dk)$. (Note that it is important here to use the soft heap implementation of \cite{KTZ13}, with an $O(1)$ amortized cost of \INSERT.)

An alternative $O(dk)$ algorithm for $d$-ary heaps, for any $d\ge 2$, can be obtained by a simple reduction from $d$-ary heaps to $3$-ary (or binary) heaps using a process that we call ``on-the-fly ternarization via heapification''. We use the well-known fact that an array of~$d$ items can be \emph{heapified}, i.e., converted into a binary heap, in $O(d)$ time. (See Williams \cite{Williams64} or Cormen et al.\ \cite{CoLeRiSt09}.) We describe this alternative approach because we think it is interesting, and because we use it in the next section to obtain an algorithm for general heap-ordered trees, i.e., trees in which different nodes may have different degrees, and the degrees of the nodes are not necessarily bounded by a constant.
}

In a $d$-ary heap, each item $e$ has (up to) $d$ children $e.\CHILD[1],\ldots, e.\CHILD[d]$. We construct a \emph{ternary} heap on the same set of items in the following way. We \emph{heapify} the $d$ children of~$e$, i.e., construct a \emph{binary} heap whose items are these $d$ children.
This gives each child~$f$ of~$e$ two new children $f.\LEFT$ and $f.\RIGHT$. (Some of these new children are null.) We let $e.\MIDDLE$ be the root of the heap composed of the children of~$e$. Overall, this gives each item~$e$ in the original heap three new children $e.\LEFT$, $e.\MIDDLE$ and $e.\RIGHT$, some of which may be null. Note that $e$ gets its new children $e.\LEFT$ and $e.\RIGHT$ when it and its siblings are heapified. (The names $\LEFT$, $\MIDDLE$ and $\RIGHT$ are, of course, arbitrary.) For an example, see Figure~\ref{F-ternarization}.

\begin{figure}[t]
\begin{center}
\includegraphics[scale=0.5]{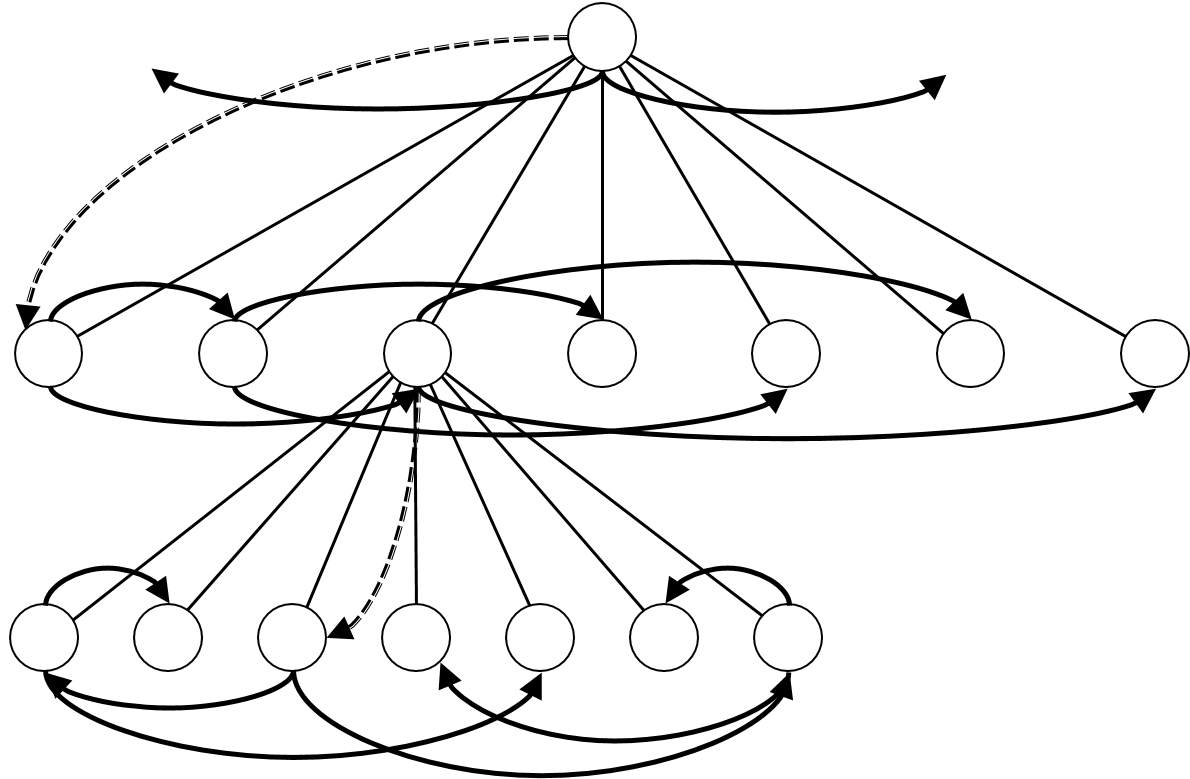}
\end{center}
\vspace*{-10pt}
\caption{On-the-fly ternarization of a 7-ary heap. Thin lines represent the original 7-ary heap. Bold arrows represent new \LEFT\ and \RIGHT\ children. Dashed arrows represent new \MIDDLE\ children.}\label{F-ternarization}
\end{figure}

This heapification process can be carried out on-the-fly while running $\SOFTSELECT(r)$ on the resulting ternary heap. The algorithm starts by inserting the root of the $d$-ary heap, which is also the root of its ternarized version, into the soft heap. When an item~$e$ is extracted from the soft heap, or becomes corrupt, we do not immediately insert its~$d$ original children  into the soft heap. Instead, we \emph{heapify} its~$d$ children, in $O(d)$ time. This assigns~$e$ its middle child $e.\MIDDLE$. Item~$e$ already has its left and right children $e.\LEFT$ and $e.\RIGHT$ defined. The three new children $e.\LEFT$, $e.\MIDDLE$ and $e.\RIGHT$ are now inserted into the soft heap. We call the resulting algorithm $\SOFTSELECTHEAPIFY(r)$.

\begin{theorem}\label{T-dary} Algorithm $\SOFTSELECTHEAPIFY(r)$ selects the $k$ smallest items from a $d$-ary heap with root~$r$ in $O(dk)$ time.
\end{theorem}

\begin{proof} Algorithm $\SOFTSELECTHEAPIFY(r)$ essentially works on a ternary version of the input $d$-ary heap constructed on the fly. Simple adaptations of Lemmas~\ref{L-soft-1} and~\ref{L-soft-2} show that the total running time, excluding the heapifications' cost, is $O(k)$. As only $O(k)$ heapifications are performed, the cost of all heapifications is $O(dk)$, giving the total running time of the algorithm.
\end{proof}

It is also possible to \emph{binarize} the input heap on the fly. We first ternarize the heap as above. We now convert the resulting ternary tree into a binary tree using the standard first child, next sibling representation. This converts the ternary heap into a binary heap, \emph{if} the three children of each item are sorted. During the ternarization process, we can easily make sure that the three children of each item appear in sorted order, swapping children if necessary, so we can apply this final binarization step.

\subsection{Selection from general heap-ordered trees}\label{S-general}

Algorithm $\SOFTSELECTHEAPIFY(r)$ works, of course, on arbitrary heap-ordered trees in which different nodes have different degrees.
Algorithm $\SOFTSELECT(r)$, on the other hand, is not easily adapted to work on general heap-ordered trees, as it is unclear how to set the error parameter $\varepsilon$ to obtain an optimal running time. To bound the running time of $\SOFTSELECTHEAPIFY(r)$ on an arbitrary heap-ordered tree, we introduce the following definition.

\begin{definition}[$D(T,k)$] Let $T$ be a (possibly infinite) rooted tree and let $k\ge 1$. Let $D(T,k)$ be the maximum \emph{sum of degrees} over all subtrees of~$T$ of size~$k$ rooted at the root of~$T$. (The degrees summed are in~$T$, not in the subtree.)
\end{definition}

For example, if $T_d$ is an infinite $d$-ary tree, then $D(T_d,k)=dk$, as the sum of degrees in each subtree of~$T_d$ containing~$k$ vertices is $dk$. For a more complicated example, let $T$ be the infinite tree in which each node at level~$i$ has degree~$i+2$, i.e., the root has two children, each of which has three children, etc. Then, $D(T,k)= \sum_{i=2}^{k+1} i = k(k+3)/2$, where the subtrees achieving this maximum are paths starting at the root.
A simple adaptation of Theorem~\ref{T-dary} gives:

\begin{theorem}\label{T-general-tree} Let $T$ be a heap-ordered tree with root~$r$. Algorithm $\SOFTSELECTHEAPIFY(r)$ selects the $k$-th smallest item in~$T$, and the set of~$k$ smallest items in~$T$, in $O(D(T,3k))$ time.
\end{theorem}

\begin{proof} We use the on-the-fly binarization and a soft heap with $\varepsilon=\frac{1}{6}$. The number of corrupt items is less than~$2k$. The number of extracted items is less than~$k$. Thus, the algorithm needs to heapify the children of less than~$3k$ items that form a subtree~$T'$ of the original tree~$T$. The sum of the degrees of these items is at most $D(T,3k)$, thus the total time spent on the heapifications, which dominates the running time of the algorithm, is $O(D(T,3k))$. We note that $D(T,3k)$ can be replaced by $D(T,(2+\delta)k)$, for any $\delta>0$, by choosing $\varepsilon$ small enough.
\end{proof}

\begin{theorem} Let $T$ be a heap-ordered tree and let $k\ge 1$. Any comparison-based algorithm for selecting the $k$-th smallest item in~$T$ must perform at least $D(T,k-1)-(k-1)$ comparisons on some inputs.
\end{theorem}

\begin{proof} Let $T'$ be the subtree of~$T$ of size~$k-1$ that achieves the value~$D(T,k-1)$, i.e., the sum of the degrees of the nodes of~$T'$ is $D(T,k-1)$. Suppose the $k-1$ items of~$T'$ are the $k-1$ smallest items in~$T$. The nodes of~$T'$ have at least $D(T,k-1)-(k-2)$ children that are not in~$T'$. The $k$-th smallest item is the minimum item among these items, and no information on the order of these items is implied by the heap order of the tree. Thus, finding the $k$-th smallest item in this case requires at least $D(T,k-1)-(k-1)$ comparisons.
\end{proof}

\section{Selection from row-sorted matrices}\label{S-matrix}

\BLUE{
In this section we present algorithms for selecting the $k$ smallest items from a row-sorted matrix, or equivalently from a collection of sorted lists. Our results simplify and extend results of Frederickson and Johnson \cite{FrJo82}. The algorithms presented in this section use our \SOFTSELECT\ algorithm for selection from a binary min-heap presented in Section~\ref{S-binary}. (Frederickson's \cite{Frederickson93} algorithm could also be used, but the resulting algorithms would become much more complicated, in particular more complicated than the algorithms of Frederickson and Johnson \cite{FrJo82}.)
}

\BLUE{
In Section~\ref{S-mk} we give an $O(m+k)$ algorithm, where $m$ is the number of rows, which is optimal if $m=O(k)$. In Section~\ref{S-mlogkm} we give an $O(m\log\frac{k+m}{m})$ algorithm which is optimal for $k=\Omega(m)$. These results match results given by Frederickson and Johnson \cite{FrJo82}. In Sections~\ref{S-Slogni} and~\ref{S-Slogki} we give two new algorithms that improve in some cases over the previous algorithms.
}

\subsection{An $O(m+k)$ algorithm}\label{S-mk}

\BLUE{
A sorted list may be viewed as a heap-sorted path, i.e., a $1$-ary heap. We can convert a collection of~$m$ sorted lists into a (degenerate) binary heap by building a binary tree whose leaves are the first items in the lists. The values of the $m-1$ internal nodes in this tree are set to $-\infty$. Each item in a list will have one real child, its successor in the list, and a dummy child with value $+\infty$. To find the~$k$ smallest items in the lists, we simply find the $m+k-1$ smallest items in the binary heap. This can be done in $O(m+k)$ time using algorithm \SOFTSELECT\ of Section~\ref{S-binary}. More directly, we can use the following straightforward modification of algorithm \SOFTSELECT. Insert the $m$ first items in the lists into a soft heap. Perform $k-1$ iterations in which an item with minimum (corrupt) key is extracted. Insert into the soft heap the child of the item extracted as well as the children of all the items that became corrupt following the \EXTRACTMIN\ operation.}

\BLUE{Alternatively, we can convert the $m$ sorted lists into a heap-ordered tree $T_{m,1}$ by adding a root with value $-\infty$ that will have the $m$ first items as its children. All other nodes in the tree will have degree~$1$.
It is easy to see that $D(T_{m,1},k)=m+k-1$. By Theorem~\ref{T-general-tree} we again get an $O(m+k)$ algorithm. We have thus presented three different proofs of the following theorem.}

\begin{theorem}\label{T-mk} Let $A$ be a row-sorted matrix containing~$m$ rows. Then, the $k$-th smallest item in~$A$, and the set of~$k$ smallest items in~$A$, can be found in $O(m+k)$ time.
\end{theorem}

We refer to the algorithm of Theorem~\ref{T-mk} as $\MATSELECT_1(A,k)$.
The $O(m+k)$ running time of $\MATSELECT_1(A,k)$ is asymptotically optimal if $k=O(m)$, as $\Omega(m)$ is clearly a lower bound; each selection algorithm must examine at least one item in each row of the input matrix.

\subsection{An $O(m\log \frac{k}{m})$ algorithm, for $k\ge 2m$}\label{S-mlogkm}

\begin{figure}[t]
\begin{center}
\includegraphics[scale=0.5]{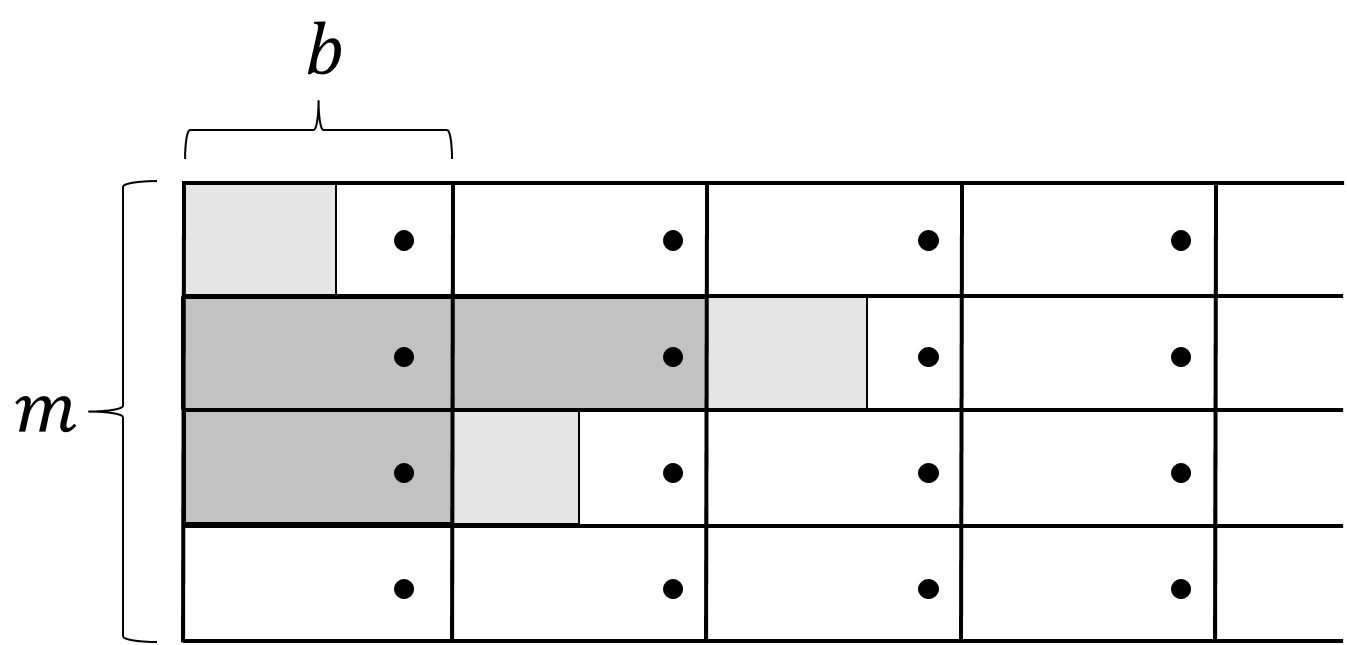}
\end{center}
\vspace*{-10pt}
\caption{Partitioning the items in each row to blocks of size~$b$. Block representatives are shown as small filled circles. The shaded regions contains the $k$ smallest items. The darkly shaded region depicts blocks all of whose items are among the $k$ smallest.}\label{F-mlogkm}
\end{figure}

\begin{figure}[t]
\begin{center}
\parbox{3.5in}{
\begin{algorithm}[H]
\DontPrintSemicolon
\NoCaptionOfAlgo
\caption{$\MATSELECT_2(A,k)$:}
$m\gets \NUMROWS(A)$ \;
\eIf{$k\le 2m$}
    {   \Return{$\MATSELECT_1(A,k)$} }
    {   $b\gets \lfloor k/(2m) \rfloor$ \;
        $A'\gets \JUMP(A,b)$ \;
        $K \gets \MATSELECT_1(A',m)$ \;
        $A''\gets \SHIFT(A,bK)$ \;
        \Return{$bK + \MATSELECT_2(A'',k-bm)$} \;
    }
\end{algorithm}

}
%
\caption{\label{F-MATSELECT2a}Selecting the~$k$ smallest items from a row-sorted matrix~$A$. (Implicit handling of submatrices passed to recursive calls.) }
\vspace*{15pt}

\parbox{4in}{
\begin{algorithm}[H]
\DontPrintSemicolon
\NoCaptionOfAlgo
\caption{$\MATSELECT_2(\langle A,c,D\rangle,k)$:}
$m\gets \NUMROWS(A)$ \;
\eIf{$k\le 2m$}
    {   \Return{$\MATSELECT_1(\langle A,c,D\rangle,k)$} }
    {   $b\gets \lfloor k/(2m) \rfloor$ \;
        $K \gets \MATSELECT_1(\langle A,bc,D\rangle,m)$ \;
        \Return{$bK + \MATSELECT_2(\langle A,c,D+bcK\rangle,k-bm)$} \;
    }
\end{algorithm}

}
\caption{\label{F-MATSELECT2b}Selecting the~$k$ smallest items from a row-sorted matrix~$A$. (Explicit handling of submatrices passed to recursive calls.)}
\end{center}
\end{figure}

We begin with a verbal description of the algorithm. Let~$A$ be the input matrix and let $k\ge 2m$.
Partition each row of the matrix~$A$ into \emph{blocks} of size $b=\left\lfloor\frac{k}{2m}\right\rfloor$. The \emph{last} item in each block is the \emph{representative} of the block.
Consider the (yet unknown) distribution of the~$k$ smallest items among the~$m$ rows of the matrix. Let $k_i$ be the number of items in the $i$-th row that are among the~$k$ smallest items in the whole matrix. These $k_i$ items are clearly the first $k_i$ items of the $i$-th row. They are partitioned into a number of full blocks, followed possibly by one partially filled block. (For an example, see Figure~\ref{F-mlogkm}.) The number of items in partially filled blocks is at most $m\left\lfloor\frac{k}{2m}\right\rfloor$. Thus, the number of filled blocks is at least
\[ \frac{ k-m\left\lfloor\frac{k}{2m}\right\rfloor } { \left\lfloor\frac{k}{2m}\right\rfloor } \;\ge\; m \;. \]
Apply algorithm $\MATSELECT_1$ to select the smallest~$m$ block representatives. This clearly takes only $O(m)$ time. (Algorithm $\MATSELECT_1$ is applied on the implicitly represented matrix~$A'$ of block representatives.)
All items in the~$m$ blocks whose representatives were selected are among the~$k$ smallest items of the matrix. The number of such items is $mb=m\left\lfloor\frac{k}{2m}\right\rfloor \ge \frac{k}{4}$, as $k\ge 2m$. These items can be removed from the matrix. All that remains is to select the $k-mb$ smallest remaining items using a recursive call to the algorithm.
In each recursive call (or iteration), the total work is $O(m)$. The number of items to be selected drops by a factor of at least $3/4$. Thus after at most $\log_{4/3}\frac{k}{2m}=O(\log\frac{k}{m})$ iterations,~$k$ drops below $2m$ and then $\MATSELECT_1$ is called to finish the job in $O(m)$ time.

Pseudo-code of the algorithm described above, which we call $\MATSELECT_2(A,k)$ is given in Figure~\ref{F-MATSELECT2a}. The algorithm returns an array $K=(k_1,k_2,\ldots,k_m)$, where $k_i$ is the number of items in the $i$-th row that are among the~$k$ smallest items of the matrix. The algorithm uses a function $\NUMROWS(A)$ that returns the number of rows of a given matrix, a function $\JUMP(A,b)$ that returns an (implicit) representation of a matrix $A'$ such that $A'_{i,j}=A_{i,bj}$, for $i,j\ge 1$, and a function $\SHIFT(A,K)$ that returns an (implicit) representation of a matrix $A''$ such that $A''_{i,j}=A_{i,j+k_i}$, for $i,j\ge 1$.

In Figure~\ref{F-MATSELECT2b} we eliminate the use of \JUMP\ and \SHIFT\ and make everything explicit. The input matrix is now represented by a triplet $\langle A,c,D\rangle$, where $A$ is a matrix, $c\ge 1$ is an integer, and $D=(d_1,d_2,\ldots,d_m)$ is
an array of non-negative integral \emph{displacements}. $\MATSELECT_2(\langle A,c,D\rangle,k)$ selects the~$k$ smallest items in the matrix $A'$ such that $A'_{i,j} = A_{i,cj+d_i}$, for $i,j\ge 1$. To select the~$k$ smallest items in~$A$ itself, we simply call $\MATSELECT_2(\langle A,1,{\bf 0}\rangle,k)$, where ${\bf 0}$ represents an array of~$m$ zeros. The implementation of $\MATSELECT_2(\langle A,c,D\rangle,k)$ in Figure~\ref{F-MATSELECT2b} is recursive. It is easy to convert it into an equivalent iterative implementation.

\begin{theorem}\label{T-mlogkm} Let $A$ be a row-sorted matrix containing~$m$ rows and let $k\ge 2m$. Algorithm $\MATSELECT_2(A,k)$ selects the~$k$ smallest items in~$A$ in $O(m\log\frac{k}{m})$ time.
\end{theorem}

Frederickson and Johnson \cite{FrJo82} showed that the $O(m\log\frac{k}{m})$ running time of $\MATSELECT_2(A,k)$ is optimal, when $k\ge 2m$. A simple proof of this claim can also be found in Section~\ref{S-lower}.


\subsection{An $O(m+\sum_{i=1}^m \log n_i)$ algorithm}\label{S-Slogni}

Assume now that the $i$-th row of~$A$ contains only $n_i$ items. We assume that $n_i\ge 1$, as otherwise, we can simply remove the $i$-th row. We can run algorithms $\MATSELECT_1$ and $\MATSELECT_2$ of the previous sections by adding dummy $+\infty$ items at the end of each row, but this may be wasteful. We now show that a simple modification of $\MATSELECT_2$, which we call $\MATSELECT_3$, can solve the selection problem in $O(m+\sum_{i=1}^m \log n_i)$ time. We focus first on the number of comparisons performed by the new algorithm.

At the beginning of each iteration, $\MATSELECT_2$ sets the block size to $b=\left\lfloor\frac{k}{2m}\right\rfloor$. If $n_i< b$, then the last item in the first block of the $i$-th row is $+\infty$. Assuming that $k\le \sum_{i=1}^m n_i$, no representatives from the $i$-th row will be selected in the current iteration. There is therefore no point in considering the $i$-th row in the current iteration.
Let $m'$ be the number of \emph{long} rows, i.e., rows for which $n_i\ge \left\lfloor\frac{k}{2m}\right\rfloor$.
We want to reduce the running time of the iteration to $O(m')$ and still reduce~$k$ by some constant factor.

The total number of items in the short rows is less than $m\left\lfloor\frac{k}{2m}\right\rfloor\le \frac{k}{2}$. The long rows thus contain at least $\frac{k}{2}$ of the~$k$ smallest items of the matrix. We can thus run an iteration of $\MATSELECT_2$ on the long rows with $k'=\frac{k}{2}$. In other words, we adjust the block size to $b'=\left\lfloor\frac{k'}{2m'}\right\rfloor=\left\lfloor\frac{k}{4m'}\right\rfloor$ and use $\MATSELECT_1$ to select the~$m'$ smallest representatives. This identifies $b'm'\ge \frac{k'}{4}\ge \frac{k}{8}$ items as belonging to the~$k$ smallest items in~$A$. Thus, each iteration takes $O(m')$ time and reduces~$k$ by a factor of at least $\frac{7}{8}$.

In how many iterations did each row of the matrix participate? Let $k_j$ be the number of items still to be selected at the beginning of iteration~$j$. Let $b_j =\left\lfloor\frac{k_j}{2m}\right\rfloor$ be the threshold for long rows used in iteration~$j$.
As $k_j$ drops exponentially, so does $b_j$. Thus, row $i$ participates in at most $O(\log n_i)$ of the \emph{last} iterations of the algorithm. The total number of comparisons performed is thus at most $O(m+\sum_{i=1}^m \log n_i)$, as claimed.

To show that the algorithm can also be implemented to run in $O(m+\sum_{i=1}^m \log n_i)$ time, we need to show that we can quickly identify the rows that are long enough to participate in each iteration. To do that, we sort $\lceil \log n_i\rceil$ using bucket sort. This takes only $O(m+\max_i \lceil \log n_i\rceil)$ time. When a row loses some of its items, it is easy to move it to the appropriate bucket in $O(1)$ time. In each iteration we may need to examine rows in one bucket that turn out not to be long enough, but this does not affect the total $O(m+\sum_{i=1}^m \log n_i)$ running time of the algorithm.

\begin{theorem}\label{T-Slogni} Let $A$ be a row-sorted matrix containing~$m$ rows, and let $N=(n_1,n_2,\ldots,n_m)$, where $n_i\ge 1$ be the number of items in the $i$-th row of the matrix, for $1\le i\le m$. Let $k\le \sum_{i=1}^m n_i$. Algorithm $\MATSELECT_3(A,N,k)$ selects the~$k$ smallest items in~$A$ in $O(m+\sum_{i=1}^m \log n_i)$ time.
\end{theorem}

In Section~\ref{S-lower} below we show that the running time of $\MATSELECT_3(A,N,k)$ is optimal for some values of $N=(n_1,n_2,\ldots,n_m)$ and $k$, e.g., if $k=\frac{1}{2}\sum_{i=1}^m n_i$, i.e., for median selection. The $O(m\log\frac{k}{m})$ running time of $\MATSELECT_2(A,k)$ is sometimes better than the $O(m+\sum_{i=1}^m \log n_i)$ running time of $\MATSELECT_3(A,N,k)$. We next describe an algorithm,  $\MATSELECT_4(A,k)$, which is always at least as fast as the three algorithms already presented, and sometimes faster.

\subsection{An $O(m+\sum_{i=1}^m \log (k_i+1))$ algorithm}\label{S-Slogki}

As before, let $k_i$ be the (yet unknown) number of items in the $i$-th row that belong to the smallest~$k$ items of the matrix. In this section we describe an algorithm for finding these $k_i$'s that runs in $O(m+\sum_{i=1}^m \log (k_i+1))$ time.

We partition each row this time into blocks of size $1,2,4,\ldots$. The representative of a block is again the last item in the block. Note that the first $k_i$ items in row $i$ reside in $\lfloor \log(k_i+1) \rfloor$ complete blocks, plus one incomplete block, if $\log(k_i+1)$ is not an integer. Thus $L=\sum_{i=1}^m \lfloor \log(k_i+1) \rfloor$ is exactly the number of block representatives that belong to the~$k$ smallest items of the matrix.

Suppose that $\ell \ge L$ is an upper bound on the true value of~$L$. We can run $\MATSELECT_1$ to select the $\ell$ smallest block representatives in $O(m+\ell)$ time. If $\ell_i$ representatives were selected from row~$i$, we let $n_i = 2^{\ell_i+1}-1$. We now run $\MATSELECT_3$ which runs in $O(m+\sum_{i=1}^m \log n_i) = O(m+\sum_{i=1}^m (\ell_i+1)) = O(m+\ell)$. Thus, if $\ell=O(L)$, the total running time is $O(m+\sum_{i=1}^m \log (k_i+1))$, as promised.

How do we find a tight upper bound on $L=\sum_{i=1}^m \lfloor \log(k_i+1) \rfloor$? We simply try $\ell = m,2m,4m,\ldots$, until we  obtain a value of~$\ell$ that is high enough. If $\ell<L$, i.e., $\ell$ is not large enough, we can discover it in one of two ways. Either $\sum_{i=1}^m n_i<k$, in which case $\ell$ is clearly too small. Otherwise, the algorithm returns an array of $k_i$ values. We can check whether these values are the correct ones in $O(m)$ time. First compute $M=\max_{i=1}^m A_{i,k_i}$. Next check that $A_{i,k_i+1}>M$, for $1\le i\le m$. As $\ell$ is doubled in each iteration, the cost of the last iteration dominates the total running time which is thus $O(m+2L) = O(m+\sum_{i=1}^m \log (k_i+1))$. We call the resulting algorithm $\MATSELECT_4$.

\begin{theorem}\label{T-Slogki} Let $A$ be a row-sorted matrix containing~$m$ rows and let $k\ge 2m$. Algorithm $\MATSELECT_4(A,k)$ selects the~$k$ smallest items in~$A$ in $O\left(m+\sum_{i=1}^m \log (k_i+1)\right)$ time, where $k_i$ is the number of items selected from row~$i$.
\end{theorem}

\subsection{Lower bounds for selection from row-sorted matrices}\label{S-lower}

We begin with a simple proof that the $O(m\log\frac{k}{m})$ algorithm is optimal for $k\ge 2m$.

\begin{theorem} Any algorithm for selecting the $k$ smallest items from a matrix with $m$ sorted rows must perform at least $(m-1)\log\frac{m+k}{m}$ comparisons on some inputs.
\end{theorem}

\begin{proof} We use the information-theoretic lower bound. We need to lower bound $s_k(m)$, which is the number of $m$-tuples $(k_1,k_2,\ldots,k_m)$, where $0\le k_i$, for $1\le i\le m$, and $\sum_{i=1}^m k_i=k$. It is easily seen that $s_k(m) = {m+k-1 \choose m-1}$, as this is the number of ways to arrange $k$ identical balls and $m-1$ identical dividers in a row. We thus get a lower bound of
\[ \log {m+k-1 \choose m-1} \;\ge\; \log \left(\frac{m+k-1}{m-1}\right)^{m-1} \;=\; (m-1)\log\frac{m+k-1}{m-1} \;\ge\; (m-1)\log\frac{m+k}{m} \;,\]
where we used the well-known relation ${n \choose k}> \left(\frac{n}{k}\right)^k$.
\end{proof}

We next show that our new $O(m+\sum_{i=1}^m\log n_i)$ algorithm is optimal, at least in some cases, e.g., when $k=\frac{1}{2}\sum_{i=1}^m n_i$ which corresponds to median selection.

\BLUE{
\begin{theorem}\label{T-lower-k} Any algorithm for selecting the $k=\frac{1}{2}\sum_{i=1}^m n_i$ smallest items from a row-sorted matrix with $m$ rows of lengths $n_1,n_2,\ldots,n_m\ge 1$ must perform at least $\sum_{i=1}^m \log (n_i+1) - \log\left(1+\sum_{i=1}^m n_i\right)$ comparisons on some inputs.
\end{theorem}

\begin{proof} The number of possible solutions to the selection problem for all values of $0\le k\le \sum_{i=1}^m n_i$ is $\prod_{i=1}^m (n_i+1)$. (Each solution corresponds to a choice $0\le k_i\le n_i$, for $i=1,2,\ldots,m$.) We prove below that the number of solutions is maximized for $k=\left\lfloor\frac{1}{2}\sum_{i=1}^m n_i\right\rfloor$ (and $k=\left\lceil\frac{1}{2}\sum_{i=1}^m n_i\right\rceil$). The number of possible solutions for this value of $k$ is thus at least
$(\prod_{i=1}^m (n_i+1))/(1+\sum_{i=1}^m n_i)$.
Taking logarithm, we get the promised lower bound.

We next prove that the number of solutions is maximized when $k=\left\lfloor\frac{1}{2}\sum_{i=1}^m n_i\right\rfloor$. Let $X_i$ be a uniform random variable on $\{0,1,\ldots,n_i\}$, and let $Y=\sum_{i=1}^m X_i$. The number of solutions for a given value~$k$ is proportional to the probability that $Y$ attains the value~$k$. Let $Y_j=\sum_{i=1}^j X_i$. We prove by induction on~$j$ that the distribution of~$Y_j$ is maximized at $\mu_j=\frac{1}{2}\sum_{i=1}^j n_i$, is symmetric around~$\mu_j$, and is increasing up to~$\mu_j$ and decreasing after~$\mu_j$. The base case is obvious as $Y_1=X_1$ is a uniform distribution. The induction step follows from an easy calculation. Indeed, $Y_j=Y_{j-1}+X_j$, where $X_j$ is uniform and $Y_{j-1}$ has the required properties. The distribution of $Y_j$ is the \emph{convolution} of the distributions of $Y_{j-1}$ and $X_j$, which corresponds to taking the \emph{average} of $n_j+1$ values of the distribution of $Y_{j-1}$. It follows easily that $Y_j$ also has the required properties.
\end{proof}

We next compare the lower bound obtained, $\sum_{i=1}^m \log (n_i+1) - \log\left(1+\sum_{i=1}^m n_i\right)$, with the upper bound $O(m+\sum_{i=1}^m \log n_i)$. The subtracted term in the lower bound is dominated by the first term, i.e., $\log\left(1+\sum_{i=1}^m n_i\right) \le \frac{\log(m+1)}{m} \sum_{i=1}^m \log (n_i+1)$, where equality holds only if $n_i=1$, for every~$i$. When the $n_i$'s are large, the subtracted term becomes negligible. Also, as $n_i\ge 1$, we have $\sum_{i=1}^m \log (n_i+1) \ge m$. Thus, the lower and upper bound are always within a constant multiplicative factor of each other.
}

The optimality of the $O(m+\sum_{i=1}^m\log n_i)$ algorithm also implies the optimality of our new ``output-sensitive" $O(m+\sum_{i=1}^m\log (k_i+1))$ algorithm. As $k_i\le n_i$, an algorithm that performs less than $c(m+ \sum_{i=1}^m \log(k_i+1))$ comparisons on all inputs, for some small enough~$c$, would contradict the lower bounds for the $O(m+\sum_{i=1}^m\log n_i)$ algorithm.

\section{Selection from $X+Y$}\label{S-XY}

We are given two \emph{unsorted} sets $X$ and $Y$ and we would like to find the $k$-th smallest item, and the set of $k$ smallest items, in the set $X+Y$. We assume that $|X|=m$, $|Y|=n$, where $m\ge n$.

\subsection{An $O(m+n+k)$ algorithm}\label{S-XY-mnk}

Heapify $X$ and heapify $Y$, which takes $O(m+n)$ time. Let $x_1,\ldots,x_m$ be the heapified order of~$X$, i.e., $x_i\le x_{2i},x_{2i+1}$, whenever the respective items exists. Similarly, let $y_1,\ldots,y_n$ be the heapified order of~$Y$. Construct a heap of maximum degree 4 representing $X+Y$ as follows.
The root is $x_1+y_1$. Item $x_i+y_1$, for $i\ge 1$ has four children $x_{2i}+y_{1},x_{2i+1}+y_{1},x_{i}+y_{2},x_{i}+y_{3}$. Item $x_i+y_j$, for $i\ge 1$, $j>1$, has two children
$x_{i}+y_{2j},x_{i}+y_{2j+1}$, again when the respective items exist. (Basically, this is a heapified version of $X+y_1$, where each $x_i+y_1$ is the root of a heapified version of $x_i+Y$.) We can now apply algorithm $\SOFTSELECT$ on this heap. We call the resulting algorithm $\XYSELECT_1(X,Y)$.

\begin{theorem}\label{T-XY-mnk} Let $X$ and $Y$ be unordered sets of $m$ and $n$ items respectively. Then, algorithm $\XYSELECT_1(X,Y)$ finds the $k$-th smallest item, and the set of $k$ smallest items in $X+Y$, in $O(m+n+k)$ time.
\end{theorem}

\subsection{An $O(m\log\frac{k}{m})$ algorithm, for $k\ge 6m$, $m\ge n$}\label{S-XY-mlogmk}

If $Y$ is sorted, then $X+Y$ is a row-sorted matrix, and we can use algorithm $\MATSELECT_2$ of Theorem~\ref{T-mlogkm}. We can sort~$Y$ in $O(n\log n)$ time and get an $O(m\log\frac{k}{m}+n\log n)$ algorithm. The running time of this algorithm is $O(m\log\frac{k}{m})$ when $k\ge mn^\varepsilon$, for any fixed $\varepsilon>0$. But, for certain values of~$m,n$ and $k$, e.g., $m=n$ and $k=mn^{o(1)}$, the cost of sorting is dominant. We show below that the sorting can always be avoided.

We first regress and describe an alternative $O(m\log\frac{k}{m})$ algorithm for selection from row-sorted matrices. The algorithm is somewhat more complicated than algorithm $\MATSELECT_2$ given in Section~\ref{S-mlogkm}. The advantage of the new algorithm is that much less assumptions are made about the order of the items in each row. A similar approach was used by Frederickson and Johnson \cite{FrJo82} but we believe that our approach is simpler. In particular we rely on a simple partitioning lemma (Lemma~\ref{L-weight} below) which is not used, explicitly or implicitly, in~\cite{FrJo82}.  

Instead of partitioning each row into blocks of equal size, as done by algorithm $\MATSELECT_2$ of Section~\ref{S-mlogkm}, we partition each row into exponentially increasing blocks, similar, but not identical, to the partition made by algorithm $\MATSELECT_3$ of Section~\ref{S-Slogni}.

Let $b=\left\lfloor\frac{k}{3m}\right\rfloor$. Partition each row into blocks of size $b,b,2b,4b,\ldots,2^jb,\ldots$. The representative of a block is again the last item in the block. We use algorithm $\MATSELECT_1$ to select the $m$ smallest representatives. This takes $O(m)$ time. Let $e_1<e_2<\ldots<e_{m}$ denote the $m$ selected representatives in (the unknown) sorted order, and let $s_1,s_2,\ldots,s_{m}$ be the sizes of their blocks. We next use an $O(m)$ \emph{weighted selection} algorithm (see, e.g., Cormen et al.\ \cite{CoLeRiSt09}, Problem 9.2, p.~225) to find the smallest $\ell$ such that $\frac{k}{6} \le \sum_{j=1}^\ell s_j$ and the items $e_1,e_2,\ldots,e_\ell$, in some order. Such an $\ell\le m$ must exist, as $m\left\lfloor\frac{k}{3m}\right\rfloor\ge m(\frac{k}{3m}-1) = \frac{1}{3}(k-3m) \ge \frac{k}{6}$, as $k\ge 6m$. Also note that $\frac{k}{6} \le \sum_{j=1}^\ell s_j < \frac{k}{3}$, as the addition of each block at most doubles the total size, i.e., $\sum_{j=1}^\ell s_j < 2\sum_{j=1}^{\ell-1} s_j$, for $\ell>1$. 

\begin{claim} All items of the blocks whose representatives are $e_1,e_2,\ldots,e_\ell$ are among the $k$ smallest items in the matrix.
\end{claim}

\begin{proof} Let $S_k$ be the set of $k$ smallest items of the matrix. Consider again the partition of $S_k$ among the $m$ rows of the matrix.  Less than $mb=m\left\lfloor\frac{k}{3m}\right\rfloor\le \frac{k}{3}$ of the items of~$S_k$ belong to rows that do not contain a full block of $S_k$ items. Thus, at least $\frac{2k}{3}$ of the items of~$S_k$ are contained in rows that contain at least one full block of $S_k$ items. The exponential increase in the size of the blocks ensures that in each such row, at least half of the items of~$S_k$ are contained in full blocks. Thus, at least $\frac{k}{3}$ of the items of~$S_k$ are contained in full blocks. In particular, if $e_1,e_2,\ldots,e_\ell$ are the smallest block representatives, and $\sum_{j=1}^\ell s_j \le \frac{k}{3}$, then all the items in the blocks of $e_1,e_2,\ldots,e_\ell$ belong to $S_k$.
\end{proof}

We can thus remove all the items in the blocks of $e_1,e_2,\ldots,e_\ell$ from the matrix and proceed to find the $k-\sum_{j=1}^\ell s_j\le \frac{5k}{6}$ smallest items of the remaining matrix. When $k$ drops below $6m$, we use the algorithm of $\MATSELECT_1$ of Section~\ref{S-mk}. The resulting algorithm performs $O(\log\frac{k}{m})$ iterations, each taking $O(m)$ time, so the total running time is $O(m\log\frac{k}{m})$. (This matches the running time of $\MATSELECT_2$, using a somewhat more complicated algorithm.)

We make another small adaptation to the new $O(m\log\frac{k}{m})$ algorithm before returning to the selection from $X+Y$ problem. Instead of letting $b=\left\lfloor\frac{k}{3m}\right\rfloor$ and using blocks of size $b,b,2b,4b,\ldots$, we let $b'=2^{\lfloor \log_2 b\rfloor}$, i.e., $b'$ is the largest power of~$2$ which is at most $b$, and use blocks of size $b',b',2b',4b',\ldots$. All block sizes are now powers of~$2$. As the sizes of the blocks may be halved, we select the $2m$ smallest block representatives. The number of items removed from each row in each iteration is now also a power of~$2$.

Back to the $X+Y$ problem. The main advantage of the new algorithm is that we do not really need the items in each row to be sorted. All we need are the items of ranks $b,b,2b,4b,\ldots$, where $b=2^\ell$ for some $\ell>0$, in each row. In the $X+Y$ problem the rows, or what remains of them after a certain number of iterations, are related, so we can easily achieve this task.

At the beginning of the first iteration, we use repeated median selection to find the items of~$Y$ whose ranks are $1,2,4,\ldots$. This also partitions $Y$ into blocks of size $1,2,4,\ldots$ such that items of each block are smaller than the items of the succeeding block. We also place the items of ranks $1,2,4,\ldots$ in their corresponding places in~$Y$. This gives us enough information to run the first iteration of the matrix selection algorithm.

In each iteration, we refine the partition of~$Y$. We apply repeated median selection on each block of size $2^\ell$ in~$Y$, breaking it into blocks of size $1,1,2,4,\ldots,2^{\ell-1}$. The total time needed is $O(n)$ per iteration, which we can easily afford. We assume for simplicity that $n=|Y|$ is a power of~$2$ and that all items in~$Y$ are distinct. We now have the following fun lemma:

\begin{figure}[t]
\begin{center}
\includegraphics[scale=0.5]{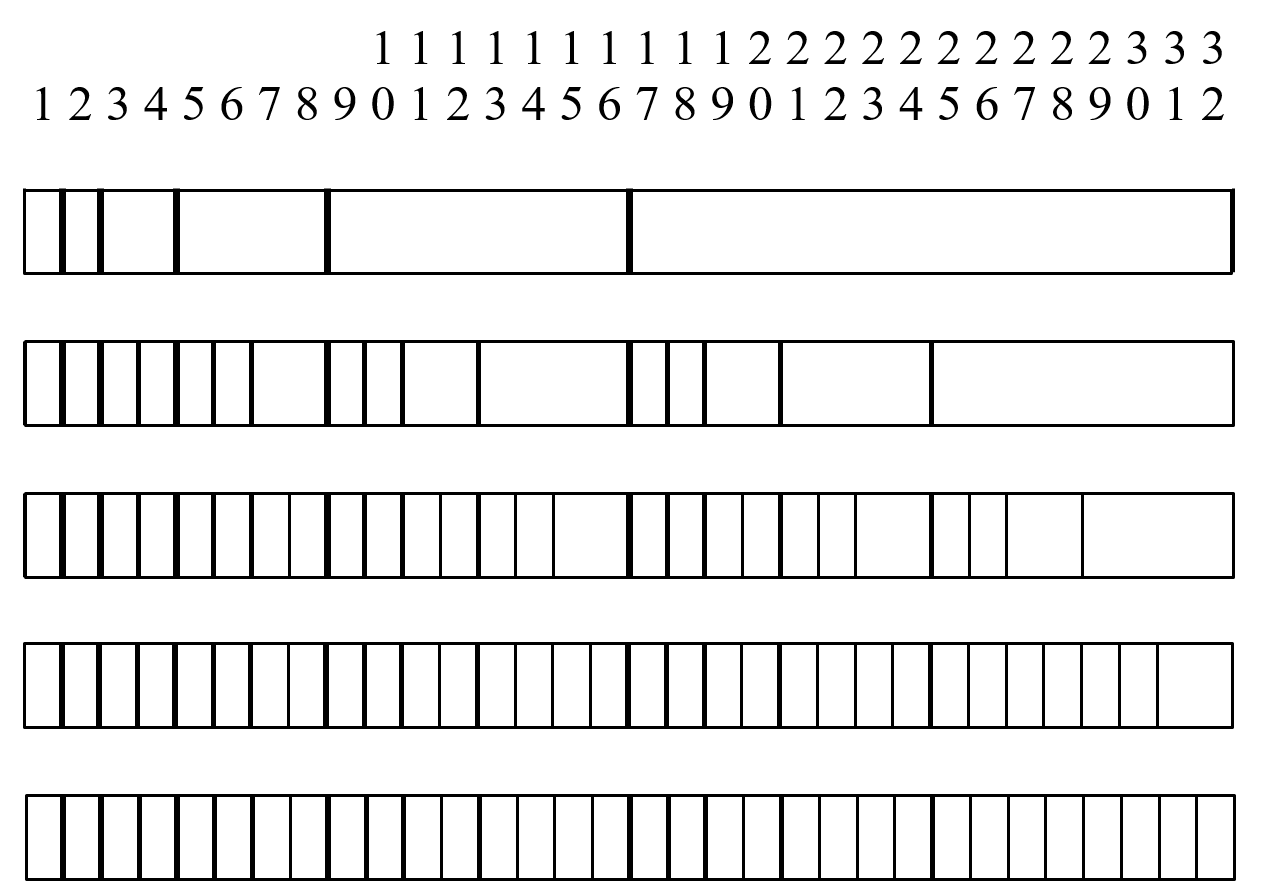}
\end{center}
\vspace*{-10pt}
\caption{Partitions of $Y$ for $n=32$.}\label{F-partition}
\end{figure}
\begin{lemma}\label{L-weight} After $i$ iterations of the above process, if $1\le r\le n$ has at most $i$ 1's in its binary representation, then $Y[r]$ is the item of rank~$r$ in~$Y$, i.e., $Y[1:r-1]<Y[r]<Y[r+1:n]$. Additionally, if $r_1<r_2$ both have at most $i$ 1's in their binary representation, and $r_2$ is the smallest number larger than~$r_1$ with this property, then $Y[1:r_1]<Y[r_1+1:r_2]<Y[r_2+1:n]$, i.e., the items in $Y[r_1+1:r_2]$ are all larger than the items in $Y[1:r_1]$ and smaller than the items in $Y[r_2+1:n]$.
\end{lemma}

For example, if $n=32$, then after the first iteration we have the partition \[Y[1],Y[2],Y[3:4],Y[5:8],Y[9:16],Y[17:32]\;.\] After the second iteration, we have the partition
\[\begin{array}{c} Y[1],Y[2],\;\;Y[3],Y[4],\;Y[5],Y[6],Y[7:8],\;\;Y[9],Y[10],\allowbreak Y[11{:}12],Y[13:16], \\Y[17],Y[18],Y[19:20],Y[21:24],Y[25:32]\;.\end{array}\] (Actually, blocks of size 2 are also sorted.) The partitions obtained for $n=32$ in the first five iterations are also shown in Figure~\ref{F-partition}.

\smallskip
\begin{proof} The claim clearly holds after the first iteration, as numbers with a single 1 in their binary representation are exactly powers of 2. Let $Y[r_1+1:r_2]$ be a block of~$Y$ generated after~$i$ iterations. If $r_1$ has less than $i$ 1's, then $r_2=r_1+1$, so the block is trivial. Suppose, therefore, that~$r_1$ has exactly~$i$ 1's in its representation and that $r_2=r_1+2^\ell$. ($\ell$ is actually the index of the rightmost~1 in the representation of~$r_1$, counting from~$0$.) In the $(i+1)$-st iteration, this block is broken into the blocks $Y[r_1+1],Y[r_1+2],Y[r_1+3:r_1+4],\ldots,Y[r_1+2^{\ell-1}+1:r_1+2^\ell]$. As the numbers $r_1+2^j$, for $1\le j<\ell$ are exactly the number between $r_1$ and $r_2$ with at most $i+1$ 1's in their binary representation, this establishes the induction step.
\end{proof}


After $i$ iterations of the modified matrix selection algorithm applied to an $X+Y$ instance, we have removed a certain number of items $d_i$ from each row. The number of items removed from each row in each iteration is a power of~$2$. By induction, $d_i$ has at most $i$ 1's in its representation. In the $(i+1)$-st iteration we set $b=2^\ell$, for some $\ell\ge 1$ and need the items of rank $b,2b,4b,\ldots$ from what remains of each row. The items needed from the $i$-th row are exactly $X[i]+Y[d_i+2^jb]$, for $j=0,1,\ldots$. The required items from~$Y$ are available, as $d_i+2^jb$ has at most $i+1$ 1s in its binary representation! We call the resulting algorithm $\XYSELECT_2(X,Y)$.

\begin{theorem}\label{T-XY-mlogkm} Let $X$ and $Y$ be unordered sets of $m$ and $n$ items respectively, where $m\ge n$, and let $k\ge 6m$. Algorithm $\XYSELECT_2(X,Y)$ finds the $k$-th smallest item, and the set of $k$ smallest items in $X+Y$, in $O(m\log\frac{k}{m})$ time.
\end{theorem}

\section{Concluding remarks}\label{S-concl}

We used soft heaps to obtain a very simple $O(k)$ algorithm for selecting the $k$-th smallest item from a binary min-heap, greatly simplifying the previous $O(k)$ algorithm of Frederickson \cite{Frederickson93}. We used this simple heap selection algorithm to obtain simpler algorithms for selection from row-sorted matrices and from $X+Y$, simplifying results of Frederickson and Johnson \cite{FrJo82}. The simplicity of our algorithms allowed us to go one step further and obtain some improved algorithms for these problems, in particular an $O(m+\sum_{i=1}^m \log(k_i+1))$ ``output-sensitive'' algorithm for selection from row-sorted matrices.

Our results also demonstrate the usefulness of soft heaps outside the realm of minimum spanning tree algorithms. It would be nice to find further applications of soft heaps.

\bibliographystyle{plain}
\bibliography{data-struc-short,selection}

\end{document}